%% file: main.tex
\newif\ifFull
\Fulltrue

\documentclass{article}

\ifFull
\usepackage[nonatbib, final]{neurips_2019}
\else
\usepackage[nonatbib, final]{neurips_2019}
\fi

\usepackage[T1]{fontenc}    
\usepackage{hyperref}       
\usepackage{url}            
\usepackage{booktabs}       
\usepackage{amsfonts}       
\usepackage{nicefrac}       
\usepackage{microtype}      
\usepackage{amsmath}
\usepackage{algorithm}
\usepackage{algorithmic}
\usepackage{caption}
\usepackage{color}
\usepackage{multirow}
\usepackage{graphicx}
\usepackage[dvipsnames]{xcolor}
\usepackage{amsthm}
\usepackage{sidecap}
\usepackage{subcaption}
\usepackage{enumitem}
\usepackage{thmtools,thm-restate}

\makeatletter
\def\BState{\State\hskip-\ALG@thistlm}
\makeatother

\newtheorem{theorem}{Theorem}[section]
\newtheorem{lemma}{Lemma}[section]
\newtheorem{definition}{Definition}[section]

\newtheorem{proposition}{Proposition}[section]

\newcommand{\obl}[1]{\tilde{#1}}
\newcommand{\dfp}[1]{\hat{#1}}
\newcommand{\oshuffle}{\mathsf{oblivious\_shuffle}}

\newcommand{\dummy}{\mathsf{dummy}}
\newcommand{\op}{\mathsf{op}}
\newcommand{\indexa}{\mathsf{a}}
\newcommand{\data}{\mathsf{data}}
\newcommand{\putop}{\mathsf{write}}
\newcommand{\getop}{\mathsf{read}}

\newcommand{\adv}[1]{\mathsf{A}_{#1}}
\newcommand{\privmem}{m}
\newcommand{\GDP}{GDP}

\newcommand{\distinct}{\mathsf{distinct}}

\newcommand{\Section}{\S}

\newlength\myindent
\newcommand\bindent{%
  \begingroup
  \setlength{\itemindent}{\myindent}
  \addtolength{\algorithmicindent}{\myindent}
}
\newcommand\eindent{\endgroup}

\newcommand{\ptr}{\mathsf{ptr}}
\newcommand{\return}{\textbf{return}~}
\newcommand{\ext}[1]{{\color{gray}#1}}

\setlist[description]{leftmargin=0.4cm,labelindent=\parindent}
\setlist[enumerate]{leftmargin=0.6cm,labelindent=\parindent}
\setlist[itemize]{leftmargin=0.6cm,labelindent=\parindent}

\makeatletter
\title{An Algorithmic Framework For Differentially Private Data Analysis on Trusted Processors}\let\Title\@title
\makeatother

\author{
\makebox[0.2\linewidth]{Joshua Allen }\\\textbf{Harsha Nori}\\
\And
\makebox[0.2\linewidth]{Bolin Ding\thanks{Current affiliation: Alibaba Group. Work done while at Microsoft.}}\\\textbf{Olga Ohrimenko}\\ \\Microsoft
\And
\makebox[0.2\linewidth]{Janardhan Kulkarni}\\\textbf{Sergey Yekhanin}\\
}

\begin{document}

\maketitle

\begin{abstract}
Differential privacy has emerged as the main definition for private data analysis and machine learning. The {\em global} model of differential privacy, which assumes that users trust the data collector, provides strong privacy guarantees and introduces small errors in the output. In contrast, applications of differential privacy in commercial systems by Apple, Google, and Microsoft, use the {\em local model}. Here, users do not trust the data collector, and hence randomize their data before sending it to the data collector. Unfortunately, local model is too strong for several important applications and hence is limited in its applicability. In this work, we propose a framework based on trusted processors and a new definition of differential privacy called {\em Oblivious Differential Privacy}, which combines the best of both local and global models. The algorithms we design in this framework show interesting interplay of ideas from the streaming algorithms, oblivious algorithms, and differential privacy. 
\end{abstract}

\input{intro}

\input{preliminaries}
\input{model}
\input{algs}

\bibliographystyle{plain}
\bibliography{main}
\ifFull\else\newpage\fi
\ifFull
\newpage
\appendix
\input{appendix}

\fi
\end{document}

%% file: intro.tex

\section{Introduction}
\label{sec:intro}
Most large IT companies rely on access to raw data from their users to train machine learning models. 
However, it is well known that models trained on a dataset can release private information about the users that participate in the dataset~\cite{DBLP:journals/corr/abs-1802-08232,DBLP:conf/sp/ShokriSSS17}. 
With new GDPR regulations and also ever increasing awareness about privacy issues in the general public, doing private and secure machine learning has become a major challenge to IT companies. 
To make matters worse, while it is easy to spot a violation of privacy when it occurs, it is much more tricky to give a rigorous definition of it.

{\em Differential privacy (DP)}, introduced in the seminal work of  Dwork {\em et al.} \cite{TCC06}, is arguably the only mathematically rigorous definition of privacy in the context of machine learning and big data analysis.
Over the past decade, DP has established itself as the defacto standard of privacy with a vast body of research and growing acceptance in industry.
Among its many strengths, the promise of DP is intuitive to explain: No matter what the adversary knows about the data, the privacy of a single user is protected from output of the data-analysis. 
A differentially private algorithm guarantees that the output does not change significantly, as quantified by a parameter $\epsilon$, if the data of any single user is omitted from the computation, which is formalized as follows.

\begin{definition} A randomized algorithm $\mathcal{A}$ is $(\epsilon,\delta)$-differentially private if for  any two neighboring databases $\mathcal{D}_1, \mathcal{D}_2$ any subset of possible outputs $S \subseteq \mathcal{Z},$ we have:
	$$\mathrm{Pr}\left[{\mathcal{A}(D_1) \in S}\right] \leq e^\epsilon\cdot \mathrm{Pr}\left[{\mathcal{A}(D_2) \in S}\right] + \delta.$$
\label{def:gdp}
\end{definition}

This above definition of DP is often called \emph{global differential privacy} (GDP). It assumes that users are willing to trust the data collector.
There is a large body of work on GDP, and many non-trivial machine learning problems can be solved in this model very efficiently. See authoritative book by Dwork and Roth ~\cite{privbook} for more details.
However in the context of IT companies, adoption of \GDP{} is not possible as there is no trusted data collector --
users want privacy of their data from the data collector.
Because of this, all industrial deployments of DP by Apple, Google, and Microsoft, with the exception of Uber~\cite{Johnson:2018:TPD:3187009.3177733},
have been set in the so called {\em local model of differential privacy (LDP)}~\cite{ccs:ErlingssonPK14,Ding2017CollectingTD,appledp}.
In the LDP model, users randomize their data {\em before} sending it
to the data collector.

\begin{definition} A randomized algorithm $\mathcal{A}: \mathcal{V} \rightarrow \mathcal{Z}$ is $\epsilon$-locally differentially private ($\epsilon$-LDP) if for any pair of values $v, v' \in \mathcal{V}$ held by a user  and any subset of output $S \subseteq \mathcal{Z},$ we have:
	$$\mathrm{Pr}\left[{\mathcal{A}(v) \in S}\right] \leq e^\epsilon\cdot \mathrm{Pr}\left[{\mathcal{A}(v') \in S}\right].$$
\label{def:ldp}
\end{definition}
\vspace{-5mm}
Despite its very strong privacy guarantees, the local model has several drawbacks compared to the global model:
many important problems cannot be solved in the LDP setting within a desired level of accuracy. Consider the simple task of understanding the number of distinct websites visited by users, or words in text data. This problem admits no good algorithms in LDP setting, whereas in the global model the problem becomes trivial. 
Even for  problems that can be solved in LDP setting \cite{ccs:ErlingssonPK14,stoc:BassilyS15,BassilyNST17,Ding2017CollectingTD}, errors and $\epsilon$ are significantly larger compared to the global model. For example, if one is interested in understanding the histogram  of websites visited by users, in the LDP setting an optimal algorithm achieves an error of $\Omega(\sqrt{n})$, whereas in the global model error is $O(\frac{1}{\epsilon})$.
See experiments and details in~\cite{prochlo} for scenarios where the errors introduced by (optimal) LDP algorithms are unacceptable in practice.
Finally, in GDP there are several results that give much stronger guarantees than the  standard composition theorems: for example, one can answer exponentially many linear queries (even online) using private multiplicative weight update algorithm \cite{Dwork:2009}.
Such results substantially increase the practical relevance of GDP algorithms.
However, the local model of differential privacy admits no such elegant solutions.
 
These drawbacks of LPD naturally lead to the following question:

{\em Are there ways to bridge the local and global differential privacy models such that users enjoy the privacy guarantees of the local model whereas the data collector enjoys the accuracy of the global model?}

This question has attracted a lot of interest in the research community recently.
In remarkable recent results, the authors of~\cite{DBLP:conf/crypto/BalleBGN19,CheuSUZZ19,ErlingssonFMRTT19} propose
a {\em secure shuffle} as a way to bridge local and global models of~DP. 
They show that if one has access to a {\em user anonymization primitive}, and if every user uses a local DP mechanism, then the overall privacy-accuracy trade-off is similar to the global model. 
However, access to anonymization primitives that users can trust is a difficult assumption to implement in practice, and only shifts the trust boundary.
For example, implementing the anonymization primitive
via mixnets requires assumption on non-collusion between the mixing servers.
Recall that the main reason most companies adopted LDP setting is because users do not trust the data collector.

In this paper, we propose a different approach based on trusted processors (for example, Intel SGX~\cite{sgx}) and a new definition called Oblivious Differential Privacy (ODP) that  help to design algorithms that enjoy the privacy guarantees of both local and global models; see~Figure~\ref{fig:ourfigs} (left) for an illustration. Our framework gives the following guarantees.
\begin{figure}
  	\centering
	\includegraphics[scale=0.33]{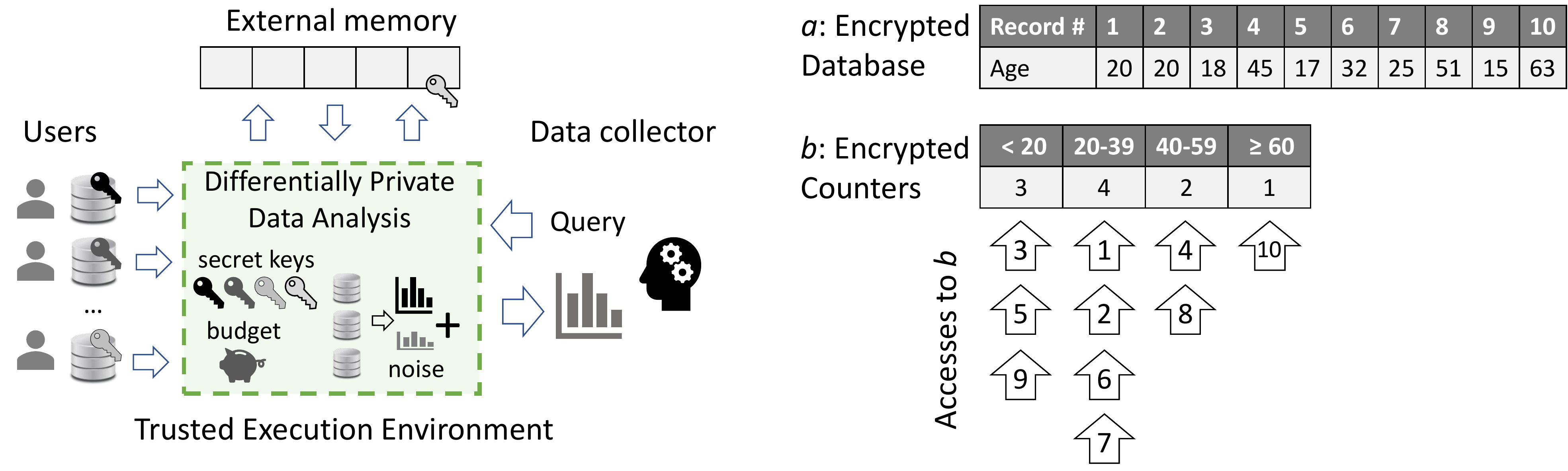}
 \vspace{-0.2cm}
 \caption{\emph{Left:} Secure differentially-private data-analysis.
  \emph{Right:} Visualization of the access pattern of a naive histogram computation over a database and four age ranges counters ($k=4$)
 stored in arrays $a$ and $b$, respectively.
 The code reads a record from $a$, decrypts it, accesses the corresponding age bucket in $b$, decrypts its counter, increments it, encrypts it
and writes it back.
The arrows indicate increment accesses to the histogram counters ($b$) and the numbers correspond to records of $a$ that were accessed prior to these accesses.
An adversary of the accesses to $a$ and $b$ learns the histogram and which database records belong to the same age range.} 
\vspace{-0.4cm}
  \label{fig:ourfigs}
\end{figure}
\begin{enumerate}
\renewcommand{\theenumi}{\Roman{enumi}}%
\item Data is collected, stored, and used in an encrypted form and is protected from the data collector.
\item The data collector obtains information about the data only through the results
of a DP-algorithm.
\end{enumerate}

The DP-algorithms themselves run within a Trusted Execution Environments (TEE) that guarantee
that the data is decrypted only by the processor during the computation and is always encrypted in memory.
Hence, raw data is inaccessible to anyone, including the data collector.
To this end, our framework is similar to other systems for private data analysis based on trusted processors,
including machine learning algorithms~\cite{Ohrimenko2016} and
data analytical platforms such as \textsc{Prochlo}~\cite{prochlo, vc3}.
{Recently, systems for supporting TEEs have been announced by Microsoft\footnote{\href{https://azure.microsoft.com/en-us/blog/introducing-azure-confidential-computing/}{``Introducing Azure confidential computing''}, accessed October 26, 2019.} and Google
\footnote{\href{https://cloudplatform.googleblog.com/2018/05/Introducing-Asylo-an-open-source-framework-for-confidential-computing.html}{``Introducing Asylo: an open-source framework for confidential computing''}, accessed October 26, 2019.}}, and we anticipate a wide adoption of this model for doing private data analysis and machine learning.

The private data analysis within trusted processors has to be done carefully since the rest
of the computational environment is untrusted and is assumed to be under adversary's control.
Though the data is encrypted, the adversary can learn private information
based on memory access patterns through caches and page faults. In particular, {\em memory access patterns} are often dependent on private information and have been shown to be sufficient
to leak information~\cite{Osvik2006,Bernstein2005, sgxsidechannels}. 
Since differentially private algorithms in the global model have been designed with a trusted data collector in mind,
they are also susceptible to information leakage through their access patterns.

The main goal of our paper is to formalize the design of differentially private
algorithms in the trusted processor environments. Our contributions are summarized below:
\begin{itemize}
\vspace{-0.12cm}\item Building on the recent works of~\cite{Ohrimenko2016,prochlo}, we propose a framework that enables collection and analysis of data in the global model of differential privacy without relying on a trusted curator. Our framework uses encryption and secure processors to protect data and computation such that
only the final differentially private output of the computation is revealed.
\vspace{-0.1cm}\item Trusted execution environments impose certain restrictions on the design of algorithms.
We formalize the mathematical model for designing differentially private algorithms in TEEs.
\vspace{-0.1cm}\item We define a new class of differentially-private algorithms (Definition \ref{def:ODP}) called \emph{obliviously differentially private algorithms} (ODP), which ensure that
privacy leakage that occurs through algorithm's memory access patterns and the output
together satisfy the DP guarantees.
\vspace{-0.1cm}\item We design ODP algorithms with provable performance guarantees for some commonly used statistical routines such as computing the number of distinct elements, histogram, and  heavy hitters. 
We prove that the privacy and error guarantees of our algorithms (Theorems (\ref{thm:distinct}, \ref{thm:histogram},\ref{thm:heavyhitters}) are significantly better than in the local model, and obliviousness does not come at a steep price.
\end{itemize}
A technically novel aspect of our paper is that it draws ideas from various different fields: streaming algorithms, oblivious algorithms, and differentially private algorithms. This fact becomes clear in~\Section4 where we design ODP algorithms.

\paragraph{Related work}
There are several systems that propose confidential data analysis using TEEs~\cite{Ohrimenko2016,opaque,vc3,prochlo}.
\textsc{Prochlo}~\cite{prochlo}, in particular, provides support for differential private data analysis.
While \textsc{Prochlo} emphasizes more on the system aspects (without formal proofs), our work gives a formal
framework based on oblivious differential privacy for analyzing and designing algorithms for private data analysis in TEEs. Oblivious sampling algorithms proposed in~\cite{odpsample}
generate samples securely s.t.~privacy amplification
can be used when analyzing DP algorithms
executed on the samples in TEE.

%% file: preliminaries.tex

\section{Preliminaries}
\label{sec:prelim}
\subsection{Secure Data Analysis with Trusted Processors and Memory Access Leakage}
\label{sec:sgx}

A visualization of our framework is given in ~Figure~\ref{fig:ourfigs}~(left). We use Intel Software Guard Extensions (SGX) as an example of a trusted processor.
Intel SGX~\cite{sgx} is a set of CPU instructions that allows user-level code to allocate a private region of memory,
called an enclave (which we also refer to as a TEE), which is accessible only to the code running in an enclave.
The enclave memory is available in raw form only inside the physical processor package, but it is encrypted and integrity protected when written to memory.
As a result, the code running inside of an enclave is isolated from the rest of the system, including the operating system.
Additionally, Intel SGX supports software attestation~\cite{sgx2}
that allows the enclave code to get messages signed with a private key of the processor along with a digest
of the enclave.
This capability allows users
to verify that they are communicating with a specific piece of software (i.e., a differentially-private algorithm) running in an enclave hosted by the trusted hardware.
Once this verification succeeds, the user can establish a secure communication channel with the enclave (e.g., using TLS)
and upload data.
When the computation is over, the enclave, including the local variables and data, is deleted.

An enclave can access data that is managed by the trusted processor (e.g., data in registers and caches) or
by the software that is not trusted (e.g., an operating system). As a result, in the latter case, data in the external memory has to be encrypted and integrity protected
by the code running inside of an enclave.
Unfortunately, encryption and integrity are not sufficient to protect against the adversary described in the introduction
that can see the addresses of the data being accessed even if the data is encrypted.
There are several ways the adversary can extract the addresses, i.e., the {\em memory access pattern}.
Some typical examples are: an adversary with physical access to a machine can attach probes to a memory bus,
an adversary that shares the same hardware as the victim enclave code (e.g., a co-tenant)
can use shared resources such as caches to observe cache-line accesses,
while a compromised operating system can inject page faults and observe page-level accesses.
Memory access patterns have been shown to be sufficient to extract secrets and data from cryptographic code~\cite{Kocher1996,Page2002,Bernstein2005,Percival2005,Osvik2006},
from genome indexing algorithms~\cite{Brasser:2017:SGE:3154768.3154779}, and from image and text applications~\cite{sgxsidechannels}.
(See Figure~\ref{fig:ourfigs} (right) for a simple example of what can be extracted by observing accesses
of a histogram computation.)
As a result, accesses leaked through memory side-channels \footnote{This should not be confused with vulnerabilities introduced by floating-point implementations~\cite{Mironov:2012:SLS:2382196.2382264}.} undermine the confidentiality promise of
enclaves~\cite{sgxsidechannels,malwareguard,sgxcacheattacks,liu2015last,sanctum,Brasser:2017:SGE:3154768.3154779,DBLP:conf/ches/MoghimiIE17}.

\subsection{Data-Oblivious Algorithms}\label{sec:obl}
Data-oblivious algorithms~\cite{GoldreichO96,Ohrimenko2016,pathoram,Goodrich:2011:DEA:1989493.1989555} are designed to protect memory addresses against the adversary described in \Section\ref{sec:sgx}:
they produce data access patterns that appear to be independent of the sensitive data
they compute on.
They can be seen as external-memory algorithms that perform computation inside of small private memory
while storing the encrypted data in the external memory and accessing it in a \emph{data-independent}
manner.
We formally capture this property below.
Suppose external memory is represented by an array $a[1,2,...,M]$ for some large value of $M$.

\begin{definition}[Access pattern]
Let $\op_j$ be either a $\getop(\indexa[i])$ operation that reads data from the location $a[i]$
to private memory or a $\putop(\indexa[i])$ operation that copies some data from the private memory to the external memory
$a[i]$. Then, let $s := (\op_1, \op_2, \ldots, \op_t)$
denote an access pattern of length $t$ of algorithm $\mathcal{A}$ to the external memory.
\end{definition}
Note that the adversary can see only the addresses accessed by the algorithm and whether it is
a read or a write. It cannot see the data since it is encrypted using
probabilistic encryption that guarantees that the adversary cannot tell if two ciphertexts correspond to the same record or two different ones.

\begin{definition}[Data-oblivious algorithm]
\label{def:obl1}

An algorithm $\mathcal{A}$ is data-oblivious if for any two inputs $I_1$ and $I_2$, and any subset of possible memory access patterns $S \subseteq \mathcal{S}$,
where $\mathcal{S}$ is the set of all possible memory access patterns produced by an algorithm, we have:
	$$\mathrm{Pr}\left[{\mathcal{A}(I_1) \in S}\right] = \mathrm{Pr}\left[{\mathcal{A}(I_2) \in S}\right]$$
\end{definition}

It is instructive to compare this definition with the definition of differential privacy. 
The definition of oblivious algorithms can be thought of as a generalization of DP to memory access patterns, where $\epsilon = 0$ and the guarantee should hold even for {\em non-neighboring} databases.
Similar to external memory algorithms, the overhead of a data-oblivious algorithm is measured in
the number of accesses it makes to external memory, while computations on private memory are assumed
to be constant. Some algorithms naturally satisfy Definition~\ref{def:obl1} while others require changes to how they operate.
For example, scanning an array is a data-oblivious algorithm since for any array of the same size every element of the array is accessed.
Sorting networks~\cite{batcher} are also data-oblivious as element identifiers accessed by compare-and-swap operations are fixed based on the size of the array and not its content.
On the other hand, quicksort is not oblivious as accesses depend on the comparison of the elements with the pivot element.
As can be seen from Figure~\ref{fig:ourfigs}~(right), a naive histogram algorithm is also not data-oblivious.
\ifFull
{(See appendix for overheads of several oblivious algorithms.)
\else{(See the supplementary material for overheads of several oblivious algorithms.)}
\fi

In this paper, we focus on measuring the
overhead in performance in terms of the number
of memory accesses of oblivious algorithms.
We omit the cost of setting up a TEE,
which is a one-time cost proportional
to the size of the code and data loaded in a TEE,
and the cost of encryption and decryption, which is linear in the size
of the data and is often implemented in hardware.

\textbf{Oblivious RAM (ORAM)} is designed to hide the indices of accesses to an array of size~$n$,
i.e., it hides how many times and when an index was accessed.
There is a naive and inefficient way to hide an access by reading and writing to every index.
Existing ORAM constructions incur sublinear overhead 
by using specialized data structures and re-arranging the external memory~\cite{GoldreichO96,GMOT12,Pinkas2010,DBLP:conf/ndss/StefanovSS12,wsc-bcomp-08,gm-paodor-11}.
The best known ORAM construction
has $O(\log n)$~\cite{cryptoeprint:2018:892} overhead.
Since it incurs high constants,
Path ORAM~\cite{pathoram} with the overhead of $O((\log n)^2)$ is a preferred option in practice.
ORAM can be used to transform any RAM program whose number of accesses does not depend on
sensitive content; otherwise, the number of accesses needs to be padded.
However, if one is willing to reveal the algorithm being performed on the data, then for some computations the overhead
of specialized constructions can be asymptotically lower than of the one based on ORAM.

\textbf{Data-oblivious shuffle~\cite{melbshuffle}}
takes as a parameter an array $a$ (stored in external memory) of size $n$ and a permutation $\pi$
and permutes $a$ according to $\pi$ such that $\pi$ is not revealed to the adversary.
The Melbourne shuffle~\cite{melbshuffle}
is a randomized data-oblivious shuffle that makes $O(cn)$ deterministic accesses to external memory, assuming private memory of size $\sqrt[c]{n}$,
and fails with negligible probability.
For example, for $c=2$ the overhead of the algorithm is constant
as any non-oblivious shuffle algorithm has to make at least $n$ accesses.
The Melbourne shuffle with smaller private memories of size $\privmem = \omega(\log n)$ incurs slightly higher
overhead of $O(n \log n / \log \privmem)$ as showed in~\cite{DBLP:journals/corr/PatelPY17}.
We will use $\oshuffle(a)$ to refer to a shuffle of $a$ according to some random permutation that is hidden from the adversary.

Note that the user anonymization primitive
that the shuffle model
of differential privacy~\cite{DBLP:conf/crypto/BalleBGN19,CheuSUZZ19,ErlingssonFMRTT19}
relies on
can be implemented in TEEs with a data-oblivious shuffle~\cite{prochlo}. 
However, in this case the trust model of the shuffle model DP will be the same as
described in the next section.

%% file: model.tex

\section{Algorithmic Framework and Oblivious Differential Privacy}
\label{sec:framework}
We now introduce the definition of Oblivious Differential Privacy (ODP), and give an algorithmic framework for the design of ODP-algorithms for a system based on trusted processors (\Section\ref{sec:sgx}).
As we mentioned earlier, off-the-shelf DP-algorithms may not be suitable for TEEs for two main reasons.

\begin{description}
	\item[Small Private Memory:] The private memory, which is protected from the adversary, available for an algorithm within a trusted processor is much smaller than the data the algorithm has to process. A reasonable assumption on the size of private memory is {\em polylogarithmic} in the input size. 
	\item[Access Patterns Leak Privacy:] An adversary who sees memory access patterns of an algorithm {to external memory} can learn useful information about the data, compromising the differential privacy guarantees of the algorithm.
\end{description}

Therefore, the algorithm designer needs to guarantee that memory access patterns do not reveal any private information, and the overall algorithm is differentially private
\footnote{Here, data collector runs algorithms on premise. See \ifFull appendix \else the supplementary material \fi for restrictions in the cloud setting.}. 
{\em To summarize, in our attacker model private information is leaked either by the output of a DP algorithm or through memory access patterns.}
We formalize this by introducing the notion of {\em Oblivious Differential Privacy}, which combines the notions of differential privacy and oblivious algorithms.

\begin{definition} 
\label{def:ODP}
Let $D_1$ and $D_2$ be any two neighboring databases that have exactly the same size $n$ but differ in one record.
A randomized algorithm $\mathcal{A}$ that has small private memory (i.e., sublinear in~$n$)
and accesses external memory
is $(\epsilon,\delta)$-obliviously differentially private (ODP), if for any subset of possible memory access patterns $S \subseteq \mathcal{S}$ and any subset of possible outputs $O$ we have:
	$$\mathrm{Pr}\left[{\mathcal{A}(D_1) \in (O,S)}\right] \leq e^\epsilon\cdot \mathrm{Pr}\left[{\mathcal{A}(D_2) \in (O, S)}\right] + \delta.$$
\end{definition}
We believe that the above definition gives a systematic way to design DP algorithms in TEE settings. 
An algorithm that satisfies the above definition guarantees that the private information released through output of the algorithm {\em and} through the access patterns is quantified by the parameters $(\epsilon, \delta)$.
Similar to our definition, Wagh~\textit{et al.}~\cite{DBLP:journals/corr/WaghCM16}
and more recently, in a parallel work, Chan~\textit{et al.}~\cite{cryptoeprint:2017:1033}
also consider relaxing the definition of obliviousness for hiding access patterns from an adversary.
However, their definitions serve complementary purpose to ours: {\em they apply DP to oblivious algorithms, whereas we apply obliviousness to DP algorithms.} 
This is crucial since algorithms that satisfy the definition in~\cite{DBLP:journals/corr/WaghCM16,cryptoeprint:2017:1033} may not satisfy DP when the output is released,
which is the main motivation for using differentially private algorithms. 
Our results together with \cite{cryptoeprint:2017:1033} highlight that DP and oblivious algorithms is an interesting area for further research for private and secure ML.

\emph{Remarks:} In the real world, the implementation of a TEE relies on cryptographic algorithms (e.g., encryption and digital signatures)
that are computationally secure and depend on a security parameter of the system.
As a result any differentially private algorithm operating inside of a TEE
has a non-zero parameter~$\delta$ that is negligible in the security parameter.

In the paper, we only focus on memory accesses;
but our definitions and framework can be easily extended to other forms of side-channel attacks such as timing attacks (e.g., by incorporating
the time of each access), albeit requiring
changes to algorithms presented in the next section
to satisfy them.
\vspace{-3mm}

\paragraph{Connections to Streaming Algorithms:}
One simple strategy to satisfy Definition \ref{def:ODP} is to take a DP-algorithm and guarantee that every time the algorithm makes an access to the public memory, it makes a pass over the entire data. 
However, such an algorithm incurs a multiplicative overhead of $n$ on the running time, and the goal would be to minimize the number of passes made over the data.
Interestingly, these algorithms precisely correspond to the {\em streaming algorithms}, which are widely studied in big-data analysis. 
In the streaming setting, one assumes that we have only $O(\log n)$ bits of memory and data stream consists of $n$ items, and the goal is to compute functions over the data.
Quite remarkably, several functions can be approximated very well in this model. 
See \cite{Muthukrishnan2003DataSA} for an extensive survey.
Since there is a large body of work on streaming algorithms, we believe that many algorithmic ideas there can be 
used in the design of ODP algorithms.
We give example of such algorithms for distinct elements \Section\ref{sec:DES} and heavy hitters problem in \Section\ref{sec:heavyhitters}.

\vspace{-3mm}
\paragraph{Privacy Budget:} Finally, we note that the system based on TEEs can support interactive data analysis where
the privacy budget is hard-coded (and hence verified by each user before they supply their private data).
The data collector's queries decrease the budget appropriately
and the code exits when privacy budget is exceeded.
Since the budget is maintained as a local variable within the TEE it is protected from replay attacks
while the code is running. If TEE exits, the adversary cannot restart it without notifying the users
since the code requires their secret keys. These keys cannot be used for different instantiations
of the same code as they are also protected by TEE and are destroyed on exit.

%% file: algs.tex
\section{Obliviously Differentially Private Algorithms}
\label{sec:odpalgos}
In this section, we show how to design ODP algorithms for the three most commonly used statistical queries: counting the number of distinct elements in a dataset, histogram of the elements, and reporting heavy hitters.
The algorithms for these problems exhibit two common themes: 1) For many applications it is possible to design DP algorithms without paying too much overhead to enforce obliviousness.
2) The interplay of ideas from the streaming and oblivious algorithms literature in the design of ODP algorithms. 

Before we continue with the construction of ODP algorithms, we make a subtle but important point. 
Recall that in our Definition \ref{def:ODP}, we require that two neighboring databases have exactly the same size. If the neighboring databases are of different sizes, then the access patterns can be of different lengths, and it is impossible to satisfy the ODP-definition. This definition does not change the privacy guarantees as in many applications the size of the database is known in advance; e.g., number of users of a system. However, it has implications on the sensitivity of the queries. For example, histogram queries in our definition have sensitivity of 2 where in the standard definition it is 1.

\subsection{Number of Distinct Items in a Database}
\label{sec:DES}
As a warm-up for the design of ODP algorithms, we begin with the distinct elements problem. 
Formally, suppose we are given a set of $n$ users and each user $i$ holds an item $v_i \in \{1,2,...,m\}$, where $m$ is assumed to be much larger than $n$. 
This is true if a company wants to understand the number of distinct websites visited by its users or the number of distinct words that occur in text data.
Let $n_v$ denote the number of users that hold the item $v$.
The goal is to estimate $n^*:= |\{ v: n_v > 0\}|$.

We are not aware of a reasonable solution
that achieves an additive error better than~$\Omega(n)$ for this problem in the LDP setting.
In a sharp contrast, the problem becomes very simple in our framework.
Indeed, a simple solution is to do an {\em oblivious sorting}~\cite{AKS,batcher} of the database elements, and then count the number of distinct elements by making another pass over the database.
Finally, one can add Laplace noise with the parameter $\frac{1}{\epsilon}$, which will guarantee that our algorithm satisfies the definition of ODP. This is true as a) the sensitivity of the query is 1, as 
a single user can increase or decrease the number of distinct elements by at most 1; b) we do oblivious sorting.
Furthermore, the expected (additive) error of such an algorithm is $1/\epsilon$. 
\ifFull
The performance of this algorithm depends on the underlying sorting algorithm:
$O(n\log n)$ with AKS~\cite{AKS} and $O(n(\log n)^2)$ with Batcher's sort~\cite{batcher}.
Though the former is asymptotically superior to Batcher's sort, it has high constants~\cite{DBLP:journals/algorithmica/Paterson90}.
\fi
Recall that $n^*$ denotes the number of distinct elements in a database. Thus we get:
\begin{theorem}
\label{thm:distinct}
There exists an oblivious sorting based $(\epsilon, 0)$-ODP algorithm for the problem of finding the number of distinct elements in a database
that runs in time $O(n \log n)$.
With probability at least $1-\theta$, the number of distinct elements output by our algorithm is $(n^* \pm \log(1/\theta) \frac{1}{\epsilon})$.
\end{theorem}
While above algorithm is optimal in terms of error, we propose  a more elegant  {\em streaming algorithm} that does the entire computation in the private memory.
The main idea is to use a {\em sketching technique} to maintain an approximate count of the distinct elements in the private memory and report this approximate count by adding noise from Lap$(1/\epsilon)$.
This will guarantee that our algorithm is $(\epsilon, 0)$-ODP, as the entire computation is done in the private memory and the Laplace mechanism is $(\epsilon, 0)$-DP. There are many streaming algorithms (e.g., Hyperloglog)~\cite{Flajolet, KaneNW10} which achieve $(1 \pm \alpha)$-approximation factor on the number of distinct elements with a space requirement of $\mathsf{polylog}(n)$.
We use the following (optimal) algorithm in \cite{KaneNW10}.
\begin{theorem}
There exists a streaming algorithm that gives a $(1 \pm \alpha)$ multiplicative approximation factor to the problem of finding the distinct elements in a data stream. The space requirement of the algorithm is at most $\frac{\log n}{\alpha^2} + (\log n)^2$ and the guarantee holds with probability $1-1/n$.
\end{theorem}
It is easy to convert the above algorithm to an ODP-algorithm by adding noise sampled from Lap$(\frac{1}{\epsilon})$.

\begin{theorem}
There exists a single pass (or online) $(\epsilon, 0)$-ODP algorithm for the problem of finding the distinct elements in a database. The space requirement of the algorithm is at most $\frac{\log n}{\alpha^2} + (\log n)^2$. With probability at least $1-1/n-\theta$, the number of distinct elements output by our algorithm is $(1 \pm \alpha)n^* \pm \log(1/\theta) \frac{1}{\epsilon}$.
\end{theorem}

The additive error of $\pm \log(1/\theta) \frac{1}{\epsilon}$ is introduced by the Laplace mechanism, and the multiplicative error  of $(1 \pm \alpha)$ is introduced by the sketching scheme. Although this algorithm is not optimal in terms of the error compared to Theorem \ref{thm:distinct}, it has the advantage that it can maintain the approximate count in an {\em online} fashion.

\subsection{Histogram}\label{sec:hist}
Let $D$ be a database with $n$ records. We assume that each record in the database has a unique identifier. 
Let $\mathcal{D}$ denote all possible databases of size $n$.
Each record (or element) $r \in \mathcal{D}$ has a {\em type}, which, without loss of generality, is an integer in the set $\{1,2, \ldots, k\}$.

For a database $D \in \mathcal{D}$, let $n_i$ denote the number of elements of type $i$.
Then the histogram function $h: \mathcal{D} \rightarrow \mathbb{R}^{k} $ is defined as 
$h(D) := (n_1, n_2, \ldots, n_k)$.

A simple differentially private histogram algorithm $\mathcal{A}_\mathsf{hist}$ returns $h(D) + (X_1, X_2, \ldots, X_k)$
where~$X_i$ are i.i.d.~random variables drawn from Lap($2/\epsilon$).
This algorithm is {\em not} obliviously differentially private as the access pattern reveals to the adversary much more information about the data than the actual output.
In this section, we design an ODP algorithm for the histogram problem.
Let $\hat{n}_i$ denote the number of elements of type $i$ output by our histogram algorithm.
We prove the following theorem in this paper.

\begin{restatable}{theorem}{histthm}
\label{thm:histogram}
For any $\epsilon >0$,  there is an $(\epsilon, \frac{1}{n^2})$-ODP algorithm for the histogram problem
that runs in time $O(\tilde{n} \log \tilde{n}/\log\log \tilde{n})$ where $\tilde{n} = \max(n, k\log n /\epsilon)$. With probability $1-\theta$, it holds that 
$$
\max_{i} |\hat{n_i} - n_{i}| \leq \log (k/\theta) \cdot \frac{2}{\epsilon}.
$$ 
\end{restatable}

Observe that our algorithm achieves same error guarantee as that of global DP algorithm without much overhead in terms of running time.

To prove that our algorithm is ODP, we need to show that 
the distribution on the access patterns produced by our algorithm for any two neighboring databases is approximately the same.
The same should hold true for the histogram output by our algorithm.
We achieve this as follows. 
We want to use the simple histogram algorithm that adds Laplace noise with parameter~$2/\epsilon$, which we know is $\epsilon$-DP. 
This is true since the histogram queries have sensitivity of~2; that is, if we change the record associated with the single user, then the histogram output changes for at most 2 types.
Note that if the private memory size is larger than~$k$, then the task is trivial.
One can build the entire DP-histogram in the private memory by making a single pass over the database, which clearly satisfies our definition of ODP.
However, in many applications $k \gg O(\log n)$. 
A common example is a histogram on {bigrams} of words which is commonly used in text prediction.
When $k \gg \log n$ the private memory is not enough to store the entire histogram, and we need to make sure that memory access patterns do not reveal too much information to the adversary.
\ifFull
Indeed, it is instructive to convince oneself that a naive implementation of the differentially private histogram algorithm in fact completely reveals the histogram to an adversary who sees the memory accesses.
\fi

One can make the naive histogram algorithm
satisfy Definition~\ref{def:ODP} by accessing the entire public memory for every read/write operation,
incurring an overhead of $O(nk)$.
\ifFull
{For~$k = \mathsf{polylog}(n)$, one can access the histogram array through an oblivious RAM. (See \ifFull appendix \else the supplementary material \fi for more details.)
Oblivious sorting algorithms could also be used to solve the histogram problem (see~\Section\ref{sec:heavyhitters} and \Section\ref{sec:otherhist}).
However, sorting is usually an expensive operation in practice.
\else
Another method to solve the histogram problem would be to rely on oblivious sorting algorithms. The overhead of this
method would be the overhead of sorting.
However, sorting is usually an expensive operation in practice.
\fi
Here, we give an arguably simpler and faster algorithm for larger values of $k$ that satisfies the definition of ODP.
{(At a high level our algorithm is similar to the one which appeared in the independent work by~Mazloom and Gordon~\cite{cryptoeprint:2017:1016},
who use a differentially private histogram to protect access patterns of graph-parallel computation based on garbled circuits,
as a result requiring a different noise distribution and shuffling algorithm.)}

\ifFull
We first  give a high-level overview of the algorithm. The pseudo-code is given in~Algorithm~\ref{alg:odphist}. 
\else
We give a high-level overview of the algorithm here, and defer the pseudo-code and analysis to the supplementary material.
\fi
Let~$T = n + 20 k \log n/\epsilon$.
\newcommand{\histogramTextAlg}{
\begin{enumerate}[itemsep=0mm]
\item Sample $k$ random variables $X_1, X_2, \ldots, X_{k}$ from Lap$(2/\epsilon)$. If any $|X_i| > 10 \log n/\epsilon$,
then we set $X_i = 0$ for all $i = 1,2,..,k$. For all $i$, set $X_i = \lceil X_i \rceil$.
\item We create $(10\log n/\epsilon + X_i)$ {\em fake} records of type $i$ and append it to the database $D$. This step together with step 1 ensures that $(10\log n/\epsilon + X_i)$ is always positive. 
The main reason to restrict the Laplace distribution's support to $[-10 \log n/\epsilon, 10 \log n/\epsilon]$ is to ensure that we only have positive noise. If the noise is negative, we cannot create fake records in the database {simulating this noise}.
\item Next, we create $(10k \log n/\epsilon - \sum_{i} X_i)$ {\em dummy} records in the database $D$, which do not correspond to any particular type in $1..k$. 
The dummy records have type $k+1$. 
The purpose of this step is to ensure that the length of the output is exactly $T$.
\item Let $\hat{D}$ be the augmented database that contains both dummy and fake records,
where the adversary cannot distinguish between database, dummy and fake records as they are encrypted using probabilistic encryption.
Obliviously shuffle $\hat{D}$ ~\cite{melbshuffle}~ so that the mapping of records to array $a[1,2,...,T]$ is uniformly distributed.
\item Initialise $b$ with $k$ zero counters in external memory.  Scan every element from the array $a[1,2,...,T]$ and increment the counter in histogram~$b$ associated with type of $a[i]$.
If the record corresponds to a dummy element, then access the array $b[1,2,...,k]$ in {\em round-robin fashion}
and do a fake write
without modifying the actual content of $b$.
\end{enumerate}
}
\histogramTextAlg

\ifFull
\input{histalg}
\input{histogramproof}
\else
In the full version \cite{corr:abs-1807-00736}, we show that above algorithm is $(\epsilon, \frac{1}{n^2})$-ODP for any $\epsilon > 0$. While the proof is long, intuition is simple: The shuffle operation gives a uniform distribution on how $\hat{D}$ is stored in public memory. Since our algorithm is deterministic after the shuffle operation -- it just makes a single pass over $\hat{D}$ -- the only information adversary learns from the memory access patterns is the histogram of the elements + noise. Now, Laplace noise guarantees that it is DP.
\fi

\subsection{Heavy Hitters}\label{sec:heavyhitters}
As a final example, we consider the problem of finding frequent items, also called the heavy hitters problem, while satisfying the ODP definition.
In this problem, we are given a database $D$ of $n$ users, where each user holds an item from the set $\{1,2,...,m\}$.
In typical applications such as finding the most frequent websites or finding the most frequently used words in a text corpus, $m$ is usually much larger than $n$. 
Hence reporting the entire histogram on $m$ items is not possible.
In such applications, one is interested in the list of $k$ most frequent items, where we define an item as frequent if it occurs more than $n/k$ times in the database.
In typical applications, $k$ is assumed to be a constant or sublinear in $n$.
The problem of finding the heavy hitters is one of the most widely studied problems in the LDP setting \cite{BassilyS15,BassilyNST17}.
In this section, we show that the heavy hitters problem becomes very simple in our model.
Let $\hat{n_i}$ denote the number of occurrences of the item $i$ output by our algorithm and $n_i$ denotes the true count. \ifFull\else We give the full proof of the theorem in the supplementary material.\fi

\begin{theorem}
\label{thm:heavyhitters}
Let $\tau > 1$ be some constant, and let $\epsilon >0$ be the privacy parameter. Suppose $n/k > \frac{\tau}{\epsilon} \log m$.  Then,  there exists a $(\epsilon, \frac{1}{m^{\tau-1}})$-ODP algorithm for the problem of finding the top $k$ most frequent items
that runs in time $O(n \log n)$.
Furthermore, for every item $i$ output by our algorithm it holds that i) with probability at least $(1-\theta)$,
$
|\hat{n_i} - n_{i}| \leq \log (m/\theta) \cdot \frac{2}{\epsilon}
$ 
and ii) $n_i \geq n/k - \log (m/\theta) \cdot \frac{2}{\epsilon}$.
\end{theorem}

We remark that the $k$ items output by an ODP-algorithm do not exactly correspond to the top $k$ items
due to the additive error introduced by the algorithm.
We can use the above theorem to return a list of {\em approximate heavy hitters}, which satisfies the following guarantees: 1) Every item
with frequency higher than~$n/k$ is in the list. 2) No item with frequency less than $n/k - 2\log (m/\theta) \cdot \frac{2}{\epsilon}$ is in the list.

We contrast the bound of this theorem with the bound one can obtain in the LDP setting. An {\em optimal} LDP algorithm can only achieve a guarantee of $
|\hat{n_i} - n_{i}| \leq \sqrt{ n \cdot \log (n/\theta) \cdot \frac{\log m}{\epsilon}}.
$ 
We refer the reader to \cite{BassilyS15, BassilyNST17} for more details about the heavy hitters problem in the LDP setting.
For many applications such as text mining, finding most frequently visited websites within a sub-population, this difference in the error turns out to be significant. See experiments and details in~\cite{prochlo}.

{
Our algorithm for Theorem~\ref{thm:heavyhitters} proceeds as follows: It sorts the elements in the database by type using oblivious sorting.
It then initialises an encrypted list~$b$ and fills it in while scanning the sorted database as follows.
It reads the first element
and saves in private memory its type, say $i$, and creates a counter set to 1.
It then appends to~$b$ a tuple: type~$i$ and the counter value.
When reading the second database element, it compares its type, say $i'$, to $i$. If $i = i'$, it increments the counter.
If $i \neq i'$, it resets the counter to 1 and overwrites the type saved in private memory to $i'$.
In both cases, it then appends to $b$ another tuple: the type and the counter from private memory.
It proceeds in this manner for the rest of the database.
Once finished, it makes a backward scan over~$b$. For every new type it encounters, it adds Lap$(2/\epsilon)$ to the corresponding counter and,
additionally, extends the tuple with a flag set to~$0$.
For all other tuples, a flag set to 1 is added instead.
It then sorts $b$: by the flag in ascending order and by differentially-private counter values in descending order.
}

{
Let $n^{*}$ be the number of distinct elements in the database.
Then the first $n^{*}$ tuples of $b$ hold all the distinct types of the database together with their differentially-private frequencies.
Since these elements are sorted, one can make a pass, returning the types of the top $k$ most frequent items with the highest count (which includes the Laplace noise).}
Although this algorithm is not $(\epsilon,0)$-ODP, it is easy to show that it is  $(\epsilon, \frac{1}{m^{\tau-1}})$-ODP when $n/k > \tau \log m$, which is the case in all applications of the heavy hitters.
Indeed, in most typical applications $k$ is a constant.
 \newcommand{\proofHeavyHitters}{
 \begin{proof}[Proof of~Theorem~\ref{thm:heavyhitters}]
{The running time of the algorithm is dominated by oblivious sorting and is $O(n \log n)$ if using AKS sorting network.
 The algorithm accesses the database and list $b$ independent of the data, hence, the accesses are also
 obliviously differentially private. In the rest of the proof, we show that the output of the algorithm is differential-private.}
 
Fix any two neighboring databases $D$ and $D'$.
Suppose $\distinct(D)$, $\distinct(D')$ denote the number of distinct items that appear in databases $D$ and $D'$. 
If $\distinct(D) = \distinct(D')$, then the distinct items appearing in
the two databases should be the same 
as $D$ and $D'$ differ only in 1 row.
{In this case, the histogram constructed on the distinct items by our algorithm is differentially private since it is an oblivious implementation
of the simple private histogram algorithm~\cite{privbook}.}
Any post-processing done on top of a differentially private output, such as reporting only top $k$ items does not violate the DP guarantee. Hence our overall algorithm is DP.

Now consider the case when $\distinct(D) \neq \distinct(D')$.
In this case, let $i^*$ denote the item that appears in $D'$ but not  in $D$.
Clearly, $n_{i^*} = 1$ in $D'$ and $n_{i^*} = 0$ in $D$.
Let $I$ denote the set of items that are common in the datasets $D$ and $D'$.
If one restricts the histogram output to the set $I$, it satisfies the guarantees of DP.
However, our algorithm never reports $i^*$ in the list of heavy hitters on the database $D$.
On the other hand, there is a non-zero probability that our algorithm reports $i^*$ in the list of top $k$ items.
This happens if the Laplace noise added to $i^*$ is greater than $n/k \geq \frac{\tau}{\epsilon} \log m$, which occurs with the probability at most $ \frac{1}{m^\tau}$. 
Since there are at most $m$ items, by the union bound the probability of this event happening is $\frac{1}{m^{\tau-1}}$.
\end{proof} 
} 
\ifFull We are now ready to prove the above theorem.
\proofHeavyHitters \else
The proof that the algorithm
satisfies the statement of Theorem~\ref{thm:heavyhitters}
appears in the supplemental material.\fi

\newcommand{\freqoracle}{
\paragraph{Frequency Oracle Based on Count-Min Sketch}
Another commonly studied problem in the context of heavy hitters is the frequency oracle problem.
Here, the goal is to answer the number of occurrences of an item $i$ in the database.
While this problem can be solved by computing the answer upon receiving a query and adding Laplace noise, there is a simpler approach which might be sufficient in many applications.
One can maintain a count-min sketch, a commonly used algorithm in the streaming literature \cite{Cormode:2005:IDS:1073713.1073718}, of the frequencies of items by making a single pass over the data.
An interesting aspect of this approach is that entire sketch can be maintained in the private memory, hence one does not need to worry about obliviousness.
Further, entire count-min sketch can be released to the data collector by adding Laplace noise.
An advantage of this approach is that the data collector can get the frequencies of any item he wants by simply referring
to the sketch, instead of consulting the DP algorithm.
It would be interesting to find more applications of the streaming techniques in the context of ODP algorithms.
}

\freqoracle

%% file: histalg.tex
\begin{figure}[t]
\begin{minipage}[t]{0.45\textwidth}
\begin{algorithm}[H]
\caption[]{Oblivious Differentially Private Histogram $\mathcal{A}^{\text{ODP}}_\mathsf{hist}(D, k)$\footnotemark}   \label{alg:odphist}
\begin{algorithmic}
\STATE $\ext{\hat{D}} \leftarrow \mathsf{add\_fake\_dummy}(\ext{D},k)${\color{white}$^\dagger$}
\STATE {$\ext{D'} \leftarrow \oshuffle(\ext{\hat{D}})$}
\STATE $\ext{b} \gets \{0\}^k${\color{white}$^\dagger$}
\STATE $\ptr \leftarrow 1${\color{white}$^\dagger$}
\FOR {$r \in \ext{D'}$}
	\STATE $i \gets \mathsf{get\_type}(r)$
	\IF {$i = k+1$} 
	\STATE  $\ext{b}_\ptr \gets \ext{b }_\ptr + 0$ {\color{white}$^\dagger$}
	\STATE $\ptr \leftarrow \ptr \mod k +1${\color{white}$^\dagger$}{\color{white}$^\dagger$}
	\ELSE
	\STATE $\ext{b}_i \gets \ext{b}_i + 1$ {\color{white}$^\dagger$}
	\ENDIF
\ENDFOR
\FOR {$i \in 1\ldots k$}
	\STATE {{ $\ext{\dfp{b}}_i \gets \ext{b}_i - 10 \log n/\epsilon$}}{\color{white}$^\dagger$}
\ENDFOR
\STATE \return $\ext{\dfp{b}}${\color{white}$^\dagger$}
\end{algorithmic}
\end{algorithm}
\end{minipage}
\hspace{1cm}
\begin{minipage}[t]{.45\textwidth}
\begin{algorithm}[H]
\caption[]{Utility procedure for Algorithm~\ref{alg:odphist}}
\begin{algorithmic}
\STATE \textbf{procedure} $\mathsf{add\_fake\_dummy}(\ext{D},k)$
\bindent
\FOR {$i \in 1\ldots k$}
	\STATE $X_i \leftarrow \textit{draw a sample from}~\text{Lap}(\frac{2}{\epsilon})$ 
\ENDFOR
\STATE \textit{If $\exists i$} $|X_i| > 10 \log n/\epsilon$: $\forall i X_i \leftarrow 0$
\STATE $\forall i, X_i \leftarrow  \lceil X_i \rceil$
\FOR {$i \in 1\ldots k$} 
	\FOR {$j \in 1\ldots (10 \log n/\epsilon + X_i $)}
		\STATE $r' \leftarrow \mathsf{set\_type}(\dummy, i)$
		\STATE $\ext{{D}}.\mathsf{append}(r')$
	\ENDFOR
\ENDFOR
\STATE $L \leftarrow 10k \cdot \log n/\epsilon - \sum^{k}_{i =1} X_i$
\FOR {$i \in 1\ldots L$}
		\STATE $r' \leftarrow \mathsf{set\_type}(\dummy, k+1)$
		\STATE $\ext{{D}}.\mathsf{append}(r')$
\ENDFOR
\STATE \return $\ext{{D}}$
\eindent

\end{algorithmic}
\end{algorithm}
\end{minipage}
\end{figure}
\footnotetext{Data stored in external memory is highlighted in \ext{grey}. Recall that it is encrypted.}

%% file: histogramproof.tex
We shall show that above algorithm is {$(\epsilon, \frac{1}{n^2})$}-differentially private for any $\epsilon > 0$. 
Towards that we need the following simple lemma for our proofs.

\begin{lemma}
\label{lem:cut}
$\Pr \left[\max_{1\le i\le k} |X_i| \geq {\frac{10 \log n}{\epsilon}} \right] \leq \frac{1}{n^2}$ where $X_i$ is drawn from~Lap$(\frac{2}{\epsilon})$.
\end{lemma}
\begin{proof}
If $Y \sim \text{Lap}(b)$, then we know that $\Pr \left[ |Y| \geq  t \cdot b \right] \leq e^{-t}.$ Therefore, 
$$\Pr \left[\max_{i} |X_i| \geq \frac{10 \log n}{\epsilon} \right] \leq \sum^k_{i = 1}\Pr \left[ |X_i| \geq \frac{10 \log n}{\epsilon} \right] \leq k \cdot 1/n^{5} \leq 1/n^2$$
where the last inequality follows from the fact that $k \leq n$.
\end{proof}

For rest of the proof, we will assume that $|X_i| \leq \frac{10 \log n}{\epsilon}$ for all $i$. 
If any $|X_i| > \frac{10 \log n}{\epsilon}$, then we will not concern ourselves in bounding the privacy loss
as it will be absorbed by the $\delta$ parameter.
Rounding of $X_i$ to a specific integer value can be seen as a post-processing step, which DP output is immune to. Hence, going forward, we will ignore this minor point.

We prove that our algorithm satisfies Definition \ref{def:ODP} in two steps: In the first step we assume that adversary does not see the access pattern, and show that output of the algorithm is {$(\epsilon, \frac{1}{n^2})$}-differentially private.

\begin{lemma}
\label{lem:laplace}
Our algorithm is $(\epsilon, \frac{1}{n^2})$-differentially private with respect to the histogram output.
\end{lemma}

\begin{proof}
We sketch the proof for completeness; see~\cite{privbook} for full proof of the lemma.
If the sensitivity of a query is $\Delta$, then we know that the Laplace mechanism, which adds a noise drawn from the distribution Lap$(\frac{\Delta}{\epsilon})$, is $(\epsilon,0)$ differentially private.
Since the histogram query has sensitivity of 2, Lap$(\frac{2}{\epsilon})$ is $(\epsilon,0)$.
However, our algorithm outputs the actual histogram without any noise if $\max_{i} |X_i| > \frac{10 \log n}{\epsilon}$, which we argued in Lemma \ref{lem:cut} happens with probability at most $\frac{1}{n^2}$.
Therefore, our algorithm is {$(\epsilon, \frac{1}{n^2})$}-differentially private.
\end{proof}

In the second step we show that memory access patterns of our algorithm satisfies oblivious differential privacy. 
That is, given any memory access pattern $s$ produced by our algorithm and two neighboring databases $D_1$ and $D_2$ we need to show
$$
\text{Pr}\left[\mathcal{A}(D_1) \rightarrow s \right] \leq e^{\epsilon} \cdot {\text{Pr} \left[\mathcal{A}(D_2) \rightarrow s \right]} + \delta 
$$
To prove this claim we need to set up some notation. 
Recall that $\hat{D}$ denotes the augmented database with dummy records and fake records.
We define {\em layout} as the actual mapping of the database $\hat{D}$ to the array $a[1], a[2], \ldots, a[T]$. 
We denote the set of all layouts by $\mathcal{L}$; clearly, $|\mathcal{L}| = T!$, where we assume that dummy/fake records also have unique identifiers.
Associated with every layout $\ell$ is a {\em configuration} $c(\ell)$. 
A configuration $c(\ell)$ is a $T$-dimensional  vector from the set $\{1,2, \ldots, k+1 \}^{T}$.
The $j$th coordinate of $c(\ell)$, denoted by $c_j(\ell)$, simply denotes a type $ i \in [k+1]$. 
Recall that dummy records are of type $k+1$.
We extend the histogram function to the augmented database in a natural fashion: $h(\hat{D}) := (n_1, n_2, \ldots, n_{k+1})$, where $n_i$ denotes the number of records of type $i$ in $h(\hat{D})$.
The shuffle operation guarantees that mapping of the database $\hat{D}$ to the array $a[1,2,\ldots, T]$ is uniformly distributed; that is, it picks a layout $\ell \in \mathcal{L}$ uniformly at random.

Given a layout $\ell \in \mathcal{L}$, the memory access pattern produced by our algorithm is completely deterministic. 
Furthermore, observe that any two layouts $\ell_1, \ell_2 \in \mathcal{L}$, which have the same configuration $c$, lead to same access pattern.
Both these statements follow directly from the definition of our algorithm, and the fact that public accesses to the array $b[1,2,\ldots, k]$ depend only on the type of the records. 

Thus, we get the following simple observation.
\begin{proposition}
\label{prep:1-2-1}
The mapping $q: \mathcal{C} \rightarrow  \mathcal{S}$ is a one-to-one mapping between the set  of all configurations $\mathcal{C}$ and the set of all access patterns produced by our algorithm $\mathcal{S}$. 
\end{proposition}

Going forward, we will concern ourselves only with the distribution produced by our algorithm on $\mathcal{C}$ rather than the set $\mathcal{S}$.
We will argue that for any two neighboring databases, the probability mass our algorithm puts on a given configuration $c \in \mathcal{C}$ satisfies the definition of DP.
The configuration produced by our algorithm depends only on two factors: a) The histogram $h(\hat{D})$ of the augmented database $\hat{D}$ b) Output of the shuffle operation.
Up to permutations, a configuration $c$ output by our algorithm is completely determined by the histogram $h(\hat{D})$ produced by our algorithm, which is a random variable.
Let $g:\mathbb{R}^{k+1} \rightarrow \{1,2,...,k+1\}^{T}$ denote the mapping from all possible histograms on augmented databases to all possible configurations of length $T$. 
However, given a layout of the database $\hat{D}$, the shuffle operation produces a random permutation of the records of the database.
This immediately implies the following lemma.

\begin{lemma}
\label{lem:config}
Fix a configuration $c \in \{1,2,...,k+1 \}^T$. Then,
$$
\frac{\mathrm{Pr}\left[{\mathcal{A}(D_1) \in c}\right]}{\mathrm{Pr}\left[{\mathcal{A}(D_2) \in c}\right]} = 
\frac{\mathrm{Pr}\left[{\mathcal{A}(D_1) \in g^{-1}(c)}\right]}{\mathrm{Pr}\left[{\mathcal{A}(D_2) \in g^{-1}(c)}\right]}
$$
\end{lemma} 

The lemma implies that to show that access patterns produced by our algorithm is $(\epsilon, \delta)$-ODP it is enough to show that distribution on the set of all histograms $\mathbb{R}^{k+1}$ satisfies  $(\epsilon, \delta)$-DP.
However, the number of dummy elements $n_{k+1}$ in any histogram $h \in \mathbb{R}^{k+1}$ output by our algorithm is completely determined by the random variables $X_1, X_2,...,X_k$. 
Hence it is enough to show that distribution on the set of all histograms on the first $k$ types satisfies  $(\epsilon, \delta)$-DP, which we already argued in Lemma is $(\epsilon, \delta)$-DP.

Now we have all the components to prove the main Theorem \ref{thm:histogram}.

\begin{proof}[Proof of~Theorem~\ref{thm:histogram}]
For any histogram $h \in \mathbb{R}^{k+1}$, let $\mathsf{truncated}(h)$ denote the histogram restricted to the first $k$ elements.
Consider any neighboring databases $D$ and $D'$ and fix a  memory access pattern $s$ produced by our algorithm. 
From Lemma \ref{lem:config} and Proposition \ref{prep:1-2-1} we have that
$$
\frac{\mathrm{Pr}\left[\mathcal{A}(D) \rightarrow s \right]}{\mathrm{Pr} \left[\mathcal{A}(D') \rightarrow s \right]} = \frac{\mathrm{Pr} \left[\mathcal{A}(D) \rightarrow g^{-1}(q^{-1}(s)) \right]}
{\mathrm{Pr}\left[\mathcal{A}(D') \rightarrow g^{-1}(q^{-1}(s)) \right]}.
$$

Since the number of dummy records is completely determined by the random variables $X_1, X_2, ..., X_k$, we have 

$$
\frac{\mathrm{Pr} \left[\mathcal{A}(D) \rightarrow g^{-1}(q^{-1}(s)) \right]}
{\mathrm{Pr}\left[\mathcal{A}(D') \rightarrow g^{-1}(q^{-1}(s)) \right]} = 
\frac{\mathrm{Pr} \left[\mathcal{A}(D) \rightarrow \mathsf{truncated}(g^{-1}(q^{-1}(s))) \right]}
{\mathrm{Pr}\left[\mathcal{A}(D') \rightarrow \mathsf{truncated}(g^{-1}(q^{-1}(s))) \right]}
$$

which is $(\epsilon, 1/n^2)$-DP from Lemma \ref{lem:laplace}.
Note that our overall mechanism is $(\epsilon, 1/n^2)$-ODP, since the memory access patterns can be used to construct the histogram output produced by our algorithm.
Therefore, we do not lose the privacy budget of $2\epsilon$.

Let us focus on showing that $\max_{i}|\hat{n_i} - n_i| \leq \log (\frac{k}{\theta}) \cdot \frac{2}{\epsilon}$ with probability at least $1-\theta$.
Consider,
\begin{eqnarray}
\Pr \left[\max_{i} |\hat{n_i} - n_{i}| \geq \log (k/\theta) \cdot \frac{2}{\epsilon} \right]  \leq  \sum^k_{i = 1} \Pr \left[ |\hat{n_i} - n_{i}| \geq \log (k/\theta) \cdot \frac{2}{\epsilon} \right]  \leq
 k \cdot \frac{\theta}{k} \leq \theta.
\end{eqnarray}

Now it only remains to bound the running time of the algorithm.
First observe that the size of the augmented database is precisely $T$. 
The shuffle operation takes $O(T \log T/\log \log T)$
and the histogram construction takes precisely $T+k$ time.
Therefore the total running time is
$O(\tilde{n} \log \tilde{n}/\log \log \tilde{n})$ where $\tilde{n} = \max(n, k \log n/\epsilon)$.
\end{proof}

%% file: appendix.tex

\section{Cloud Computing Setting}
\label{sec:cloud}
If the framework is hosted on the cloud, we also consider a second adversary $\adv{C}$
who is hosting the framework
(e.g., malicious cloud administrator or co-tenants).
Since $\adv{C}$~has access to the infrastructure of the framework its adversarial capabilities are the same as those of
the adversarial data collector $\adv{D}$ in \Section\ref{sec:intro}.
However, in this case, in addition to protecting user data from $\adv{D}$, the result of the computation also needs to
be protected.
We note that if $\adv{D}$ and $\adv{C}$ are colluding,
then the cloud scenario is equivalent to the on-premise one.
If they are not colluding, $\adv{C}$ may still be able to gain access to the framework
where its outsourced computation is performed. It can do so either by
using malware or as a malicious co-tenant~\cite{Ristenpart:2009:HYG:1653662.1653687}.

In this setting, our framework remains as described in \Section\ref{sec:framework}
with the exception that algorithms that run inside of the TEE need to be data-oblivious per Definition~\ref{def:obl1}
in order to hide the output from~$\adv{C}$. Since Definition~\ref{def:obl1} is stronger than Definition~\ref{def:ODP}
when output is not revealed, user data is also protected from~$\adv{D}$.

\section{ORAM-based Histogram Algorithm}
\label{sec:otherhist}

In this section we outline a standard differentially private algorithm for histogram computation
and its ORAM-based transformation to achieve ODP. We note that the algorithm in \Section\ref{sec:hist}
is more efficient than this transformation for large values of~$k$.

Algorithm~\ref{alg:hist} is the standard differentially private algorithm from~\cite{privbook}.
It is not data-oblivious since accesses to the histogram depend on the content of the database and reveal which records
have the same type.
(See~Figure~\ref{fig:ourfigs} (right) for a visualization.)

Algorithm~\ref{alg:ohist1} is a data-oblivious version of Algorithm~\ref{alg:hist}.
It uses oblivious RAM (see \Section\ref{sec:obl}) to hide accesses to the histogram.
In particular, we use $\mathsf{ORAM}(b)$ to denote the
algorithm that returns a data-oblivious data structure $\obl{b}$
initialized with an array $b$. $\obl{b}$ supports queries $\obl{b}.\getop(i)$
and $\obl{b}.\putop(i, \data_i)$ for $i \in [1,k]$. That is, it returns $b.\getop(i)$ (or similarly stores $\data_i$ in the $i$th index of $b$ with $\putop$)
while hiding~$i$.
The performance of the resulting algorithm depends on the underlying oblivious RAM construction.
For example, if we use the scheme by~Asharov~\textit{et al.}~\cite{cryptoeprint:2018:892},
it makes $O(n \log k)$ accesses to external memory,
or $O(n (\log k)^2)$ if we use Path ORAM~\cite{pathoram} that has much smaller constants.

The histogram problem can be solved also using two oblivious sorts, similar to the heavy hitters algorithm
described in~\Section\ref{sec:heavyhitters} but returning all $n^{*}$ types and their counters
since $k=n^{*}$ when histogram support is known.
The overhead of this method is the overhead of the underlying sorting algorithm.

\begin{figure}[h]
\begin{minipage}[t]{.45\textwidth}
\begin{algorithm}[H]
\begin{algorithmic}
	\caption[]{DP histogram~\cite{privbook} $\mathcal{M}^{\mathsf{DP}}_\mathsf{hist}(D, k)$}
	\label{alg:hist}
\STATE $\ext{b} \gets \{0\}^k$
\STATE {\color{white} ${\obl{b}} $}
\FOR {$r \in \ext{D}$}
	\STATE $i \gets \mathsf{get\_type}(r)$
	\STATE $\ext{b}_i \gets \ext{b}_i + 1$  {\color{white} ${\obl{b}} $} {\color{white} $b_i \leftarrow {\obl{b}}.\getop(i)$}
	\STATE {\color{white} ${\obl{b}}.\putop(i, b_i+1)$}
\ENDFOR
\FOR {$i \in 1\ldots k$}
	\STATE $\ext{\dfp{b}}_i \gets \ext{b}_i + \text{Lap}(\frac{2}{\epsilon})$ {\color{white} ${\obl{b}} $}
\ENDFOR
\STATE return $\ext{\dfp{b}}$
\end{algorithmic}
\end{algorithm}
\end{minipage}%
\hspace{.4cm}
\begin{minipage}[t]{.50\textwidth}
\begin{algorithm}[H]
\begin{algorithmic}
    \caption{ORAM-based DP histogram $\mathcal{M}^\mathsf{ODP}_\mathsf{hist}(D, k)$}
    \label{alg:ohist1}
\STATE $b \gets \{0\}^k$
\STATE {\color{blue} $\ext{\obl{b}}  \gets \mathsf{ORAM}(c)$}
\FOR {$r \in \ext{D}$}
	\STATE $i \gets \mathsf{get\_type}(r)$
	\STATE {\color{blue} $b_i \leftarrow \ext{\obl{b}}.\getop(i)$}
	\STATE {\color{blue} $\ext{\obl{b}}.\putop(i, b_i+1)$}
\ENDFOR
\FOR {$i \in 1\ldots k$}
	\STATE {\color{blue} $\ext{\dfp{b}}_i \gets \ext{\obl{b}}.\getop(i) + \text{Lap}(\frac{2}{\epsilon})$}
\ENDFOR
\STATE return $\ext{\dfp{b}}$
\end{algorithmic}
\end{algorithm}
\end{minipage}
\end{figure}

\newpage

\section{Comparison of Data-Oblivious Algorithms}
\label{sec:oramtbl}
\begin{table}[h]
    \renewcommand{\arraystretch}{1.5}
\begin{center}
    \caption{Asymptotical performance of several data-oblivious algorithms on arrays of $n$ records, $c \ge 2$ is a constant,
     $\privmem$ is the size of private memory, $k$ is the number of types in a histogram, and $\tilde{n} = \max(n, k \log n/\epsilon)$.}\label{tbl:oblcmp}
     \vspace{2mm}
    \begin{tabular}{|l |  l | c | c |}
    \hline
    & \textbf{Algorithm} & \textbf{Private Memory ($\privmem$)} & {\textbf{Overhead}}\\
    \hline
    \multirow{3}{*}{RAM } & \multirow{1}{*}{Path ORAM~\cite{pathoram}} & $\omega(\log n)$& $(\log n)^2$ \\
				 	    & {Kushilevitz~\textit{et al.}~\cite{KLO}} & $1$ & $\frac{(\log n)^2}{\log \log n}$ \\
				 	     & {Asharov~\textit{et al.}~\cite{cryptoeprint:2018:892}} & $1$ & $\log n$ \\
    \hline
    \hline
    \multirow{2}{*}{Sort} & AKS Sort~\cite{AKS} & 1 & $n \log n$ \\
    & Batcher's Sort~\cite{batcher} & 1 & $n (\log n)^2$ \\
    \hline
    \hline    
    \multirow{2}{*}{Shuffle} & \multirow{2}{*}{Melbourne Shuffle~\cite{melbshuffle,DBLP:journals/corr/PatelPY17}} & $\sqrt[c]{n}$ & $cn$ \\
    								\cline{3-4}
						    		& & $\omega(\log n)$ & $n\frac{\log n}{\log \privmem}$ \\
    \hline
    \hline
    \multirow{3}{*}{Histogram} 
	    & In private memory,~{$k = O(\privmem)$} &  $k$ & $n$ \\
   	    & ORAM-based (Algorithm~\ref{alg:ohist1})&  $1$ & $n{\log k}$ \\
      	    & Oblivious Sort-based (\Section\ref{sec:otherhist})&  $1$ & $n\log n$ \\
      	    \cline{2-4}
	    & {{ODP Histogram (\Section\ref{sec:hist}, Algorithm~\ref{alg:odphist})}} &  $\omega(\log n)$ & $\tilde{n} \frac{\log \tilde{n}}{\log \log \tilde{n}}$ \\
	       \hline
    \end{tabular}
\end{center}
\end{table}

%% file: main.bbl
\begin{thebibliography}{10}

\bibitem{AKS}
Mikl\'{o}s Ajtai, J\'{a}nos Koml\'{o}s, and Endre Szemer{\'e}di.
\newblock {An $O(n \log n)$ sorting network}.
\newblock In {\em ACM Symposium on Theory of Computing (STOC)}, 1983.

\bibitem{sgx2}
Ittai Anati, Shay Gueron, Simon Johnson, and Vincent Scarlata.
\newblock Innovative technology for {CPU} based attestation and sealing.
\newblock In {\em Workshop on Hardware and Architectural Support for Security
  and Privacy (HASP)}, 2013.

\bibitem{cryptoeprint:2018:892}
Gilad Asharov, Ilan Komargodski, Wei-Kai Lin, Kartik Nayak, Enoch Peserico, and
  Elaine Shi.
\newblock Optorama: Optimal oblivious ram.
\newblock Cryptology ePrint Archive, Report 2018/892, 2018.
\newblock \url{https://eprint.iacr.org/2018/892}.

\bibitem{DBLP:conf/crypto/BalleBGN19}
Borja Balle, James Bell, Adri{\`{a}} Gasc{\'{o}}n, and Kobbi Nissim.
\newblock The privacy blanket of the shuffle model.
\newblock In {\em Advances in Cryptology---CRYPTO}, pages 638--667, 2019.

\bibitem{BassilyNST17}
Raef Bassily, Kobbi Nissim, Uri Stemmer, and Abhradeep~Guha Thakurta.
\newblock Practical locally private heavy hitters.
\newblock In {\em Conference on Neural Information Processing Systems
  (NeurIPS)}, pages 2285--2293, 2017.

\bibitem{stoc:BassilyS15}
Raef Bassily and Adam~D. Smith.
\newblock Local, private, efficient protocols for succinct histograms.
\newblock In {\em ACM Symposium on Theory of Computing (STOC)}, pages 127--135,
  2015.

\bibitem{BassilyS15}
Raef Bassily and Adam~D. Smith.
\newblock Local, private, efficient protocols for succinct histograms.
\newblock In {\em ACM Symposium on Theory of Computing (STOC)}, pages 127--135,
  2015.

\bibitem{batcher}
Kenneth~E. Batcher.
\newblock Sorting networks and their applications.
\newblock In {\em Spring Joint Computer Conf.}, 1968.

\bibitem{Bernstein2005}
Daniel~J. Bernstein.
\newblock Cache-timing attacks on {AES}.
\newblock Technical report, Department of Mathematics, Statistics, and Computer
  Science, University of Illinois at Chicago, 2005.

\bibitem{prochlo}
Andrea Bittau, \'{U}lfar Erlingsson, Petros Maniatis, Ilya Mironov, Ananth
  Raghunathan, David Lie, Mitch Rudominer, Ushasree Kode, Julien Tinnes, and
  Bernhard Seefeld.
\newblock Prochlo: Strong privacy for analytics in the crowd.
\newblock In {\em ACM Symposium on Operating Systems Principles (SOSP)}, 2017.

\bibitem{Brasser:2017:SGE:3154768.3154779}
Ferdinand Brasser, Urs M\"{u}ller, Alexandra Dmitrienko, Kari Kostiainen,
  Srdjan Capkun, and Ahmad-Reza Sadeghi.
\newblock Software grand exposure: {SGX} cache attacks are practical.
\newblock In {\em USENIX Workshop on Offensive Technologies (WOOT)}, 2017.

\bibitem{DBLP:journals/corr/abs-1802-08232}
Nicholas Carlini, Chang Liu, \'{U}lfar Erlingsson, Jernej Kos, and Dawn Song.
\newblock {The Secret Sharer: Evaluating and Testing Unintended Memorization in
  Neural Networks}.
\newblock In {\em USENIX Security Symposium}, 2019.

\bibitem{cryptoeprint:2017:1033}
T-H.~Hubert Chan, Kai-Min Chung, Bruce~M. Maggs, and Elaine Shi.
\newblock Foundations of differentially oblivious algorithms.
\newblock In {\em ACM-SIAM Symposium on Discrete Algorithms (SODA)}, pages
  2448--2467, 2019.

\bibitem{CheuSUZZ19}
Albert Cheu, Adam Smith, Jonathan Ullman, David Zeber, and Maxim Zhilyaev.
\newblock Distributed differential privacy via shuffling.
\newblock In {\em Advances in Cryptology-–-EUROCRYPT}, 2019.

\bibitem{Cormode:2005:IDS:1073713.1073718}
Graham Cormode and S.~Muthukrishnan.
\newblock An improved data stream summary: The count-min sketch and its
  applications.
\newblock {\em J. Algorithms}, 55(1):58--75, April 2005.

\bibitem{sanctum}
Victor Costan, Ilia Lebedev, and Srinivas Devadas.
\newblock Sanctum: Minimal hardware extensions for strong software isolation.
\newblock In {\em USENIX Security Symposium}, 2016.

\bibitem{appledp}
Apple Differential Privacy~Team.
\newblock Learning with privacy at scale, 2017.

\bibitem{Ding2017CollectingTD}
Bolin Ding, Janardhan Kulkarni, and Sergey Yekhanin.
\newblock Collecting telemetry data privately.
\newblock In {\em Conference on Neural Information Processing Systems
  (NeurIPS)}, 2017.

\bibitem{TCC06}
Cynthia Dwork, Frank McSherry, Kobbi Nissim, and Adam Smith.
\newblock Calibrating noise to sensitivity in private data analysis.
\newblock In {\em Theory of Cryptography Conference (TCC)}, pages 265--284,
  2006.

\bibitem{Dwork:2009}
Cynthia Dwork, Moni Naor, Omer Reingold, Guy~N. Rothblum, and Salil Vadhan.
\newblock On the complexity of differentially private data release: Efficient
  algorithms and hardness results.
\newblock In {\em ACM Symposium on Theory of Computing (STOC)}, pages 381--390,
  2009.

\bibitem{privbook}
Cynthia Dwork and Aaron Roth.
\newblock The algorithmic foundations of differential privacy.
\newblock {\em Found. Trends Theor. Comput. Sci.}, 9, August 2014.

\bibitem{ErlingssonFMRTT19}
\'{U}lfar Erlingsson, Vitaly Feldman, Ilya Mironov, Ananth Raghunathan, Kunal
  Talwar, and Abhradeep Thakurta.
\newblock Amplification by shuffling: From local to central differential
  privacy via anonymity.
\newblock In {\em ACM-SIAM Symposium on Discrete Algorithms (SODA)}, 2019.

\bibitem{ccs:ErlingssonPK14}
{\'{U}}lfar Erlingsson, Vasyl Pihur, and Aleksandra Korolova.
\newblock {RAPPOR:} randomized aggregatable privacy-preserving ordinal
  response.
\newblock In {\em ACM Conference on Computer and Communications Security
  (CCS)}, pages 1054--1067, 2014.

\bibitem{Flajolet}
Philippe Flajolet and G.~Nigel Martin.
\newblock Probabilistic counting algorithms for data base applications.
\newblock {\em J. Comput. Syst. Sci.}, 31(2), September 1985.

\bibitem{GoldreichO96}
Oded Goldreich and Rafail Ostrovsky.
\newblock Software protection and simulation on oblivious {RAM}s.
\newblock {\em Journal of the ACM (JACM)}, 43(3), 1996.

\bibitem{Goodrich:2011:DEA:1989493.1989555}
Michael~T. Goodrich.
\newblock Data-oblivious external-memory algorithms for the compaction,
  selection, and sorting of outsourced data.
\newblock In {\em Proceedings of the Twenty-third Annual ACM Symposium on
  Parallelism in Algorithms and Architectures}, 2011.

\bibitem{gm-paodor-11}
Michael~T. Goodrich and Michael Mitzenmacher.
\newblock Privacy-preserving access of outsourced data via oblivious {RAM}
  simulation.
\newblock In {\em International Colloquium on Automata, Languages and
  Programming (ICALP)}, 2011.

\bibitem{GMOT12}
Michael~T. Goodrich, Michael Mitzenmacher, Olga Ohrimenko, and Roberto
  Tamassia.
\newblock Privacy-preserving group data access via stateless oblivious {RAM}
  simulation.
\newblock In {\em ACM-SIAM Symposium on Discrete Algorithms (SODA)}, 2012.

\bibitem{sgxcacheattacks}
Johannes G\"{o}tzfried, Moritz Eckert, Sebastian Schinzel, and Tilo M\"{u}ller.
\newblock Cache attacks on {Intel SGX}.
\newblock In {\em European Workshop on System Security (EuroSec)}, 2017.

\bibitem{sgx}
Matthew Hoekstra, Reshma Lal, Pradeep Pappachan, Carlos Rozas, Vinay Phegade,
  and Juan del Cuvillo.
\newblock Using innovative instructions to create trustworthy software
  solutions.
\newblock In {\em Workshop on Hardware and Architectural Support for Security
  and Privacy (HASP)}, 2013.

\bibitem{Johnson:2018:TPD:3187009.3177733}
Noah Johnson, Joseph~P. Near, and Dawn Song.
\newblock Towards practical differential privacy for {SQL} queries.
\newblock {\em PVLDB}, 11(5):526--539, January 2018.

\bibitem{KaneNW10}
Daniel~M. Kane, Jelani Nelson, and David~P. Woodruff.
\newblock An optimal algorithm for the distinct elements problem.
\newblock In {\em Symposium on Principles of Database Systems (PODS)}, pages
  41--52, 2010.

\bibitem{Kocher1996}
Paul~C. Kocher.
\newblock Timing attacks on implementations of {Diffe-Hellman, RSA, DSS}, and
  other systems.
\newblock In {\em Advances in Cryptology---CRYPTO}, 1996.

\bibitem{KLO}
Eyal Kushilevitz, Steve Lu, and Rafail Ostrovsky.
\newblock On the (in)security of hash-based oblivious {RAM} and a new balancing
  scheme.
\newblock In {\em ACM-SIAM Symposium on Discrete Algorithms (SODA)}, 2012.

\bibitem{liu2015last}
Fangfei Liu, Yuval Yarom, Qian Ge, Gernot Heiser, and Ruby~B Lee.
\newblock Last-level cache side-channel attacks are practical.
\newblock In {\em IEEE Symposium on Security and Privacy (S\&P)}, 2015.

\bibitem{cryptoeprint:2017:1016}
Sahar Mazloom and S.~Dov Gordon.
\newblock Secure computation with differentially private access patterns.
\newblock In {\em ACM Conference on Computer and Communications Security
  (CCS)}, 2018.

\bibitem{Mironov:2012:SLS:2382196.2382264}
Ilya Mironov.
\newblock On significance of the least significant bits for differential
  privacy.
\newblock In {\em ACM Conference on Computer and Communications Security
  (CCS)}, 2012.

\bibitem{DBLP:conf/ches/MoghimiIE17}
Ahmad Moghimi, Gorka Irazoqui, and Thomas Eisenbarth.
\newblock Cachezoom: How {SGX} amplifies the power of cache attacks.
\newblock In {\em Cryptographic Hardware and Embedded Systems (CHES)}, 2017.

\bibitem{Muthukrishnan2003DataSA}
S.~Muthukrishnan.
\newblock Data streams: algorithms and applications.
\newblock {\em Foundations and Trends in Theoretical Computer Science}, 1,
  2003.

\bibitem{melbshuffle}
Olga Ohrimenko, Michael~T. Goodrich, Roberto Tamassia, and Eli Upfal.
\newblock The {Melbourne} shuffle: Improving oblivious storage in the cloud.
\newblock In {\em International Colloquium on Automata, Languages and
  Programming (ICALP)}, volume 8573. Springer, 2014.

\bibitem{Ohrimenko2016}
Olga Ohrimenko, Felix Schuster, C\'edric Fournet, Aastha Mehta, Sebastian
  Nowozin, Kapil Vaswani, and Manuel Costa.
\newblock Oblivious multi-party machine learning on trusted processors.
\newblock In {\em USENIX Security Symposium}, 2016.

\bibitem{Osvik2006}
Dag~Arne Osvik, Adi Shamir, and Eran Tromer.
\newblock Cache attacks and countermeasures: the case of {AES}.
\newblock In {\em RSA Conference Cryptographer's Track (CT-RSA)}, 2006.

\bibitem{Page2002}
Dan Page.
\newblock {Theoretical use of cache memory as a cryptanalytic side-channel}.
\newblock {\em Cryptology ePrint Archive, Report 2002/169}, 2002.

\bibitem{DBLP:journals/corr/PatelPY17}
Sarvar Patel, Giuseppe Persiano, and Kevin Yeo.
\newblock Cacheshuffle: {A} family of oblivious shuffles.
\newblock In {\em International Colloquium on Automata, Languages and
  Programming (ICALP)}, 2018.

\bibitem{DBLP:journals/algorithmica/Paterson90}
Mike Paterson.
\newblock {Improved Sorting Networks with $O(\log N)$ Depth}.
\newblock {\em Algorithmica}, 5(1):65--92, 1990.

\bibitem{Percival2005}
Colin Percival.
\newblock {Cache missing for fun and profit}.
\newblock In {\em Proceedings of BSDCan}, 2005.

\bibitem{Pinkas2010}
Benny Pinkas and Tzachy Reinman.
\newblock Oblivious {RAM} revisited.
\newblock In {\em Advances in Cryptology---CRYPTO}, 2010.

\bibitem{Ristenpart:2009:HYG:1653662.1653687}
Thomas Ristenpart, Eran Tromer, Hovav Shacham, and Stefan Savage.
\newblock Hey, you, get off of my cloud: Exploring information leakage in
  third-party compute clouds.
\newblock In {\em ACM Conference on Computer and Communications Security
  (CCS)}, 2009.

\bibitem{odpsample}
Sajin Sasy and Olga Ohrimenko.
\newblock Oblivious sampling algorithms for private data analysis.
\newblock In {\em Conference on Neural Information Processing Systems
  (NeurIPS)}, 2019.

\bibitem{vc3}
Felix Schuster, Manuel Costa, C{\'e}dric Fournet, Christos Gkantsidis, Marcus
  Peinado, Gloria Mainar-Ruiz, and Mark Russinovich.
\newblock {$VC3$}: Trustworthy data analytics in the cloud using {SGX}.
\newblock In {\em IEEE Symposium on Security and Privacy (S\&P)}, 2015.

\bibitem{malwareguard}
Michael Schwarz, Samuel Weiser, Daniel Gruss, Clementine Maurice, and Stefan
  Mangard.
\newblock {Malware Guard Extension}: Using {SGX} to conceal cache attacks.
\newblock In {\em Conference on Detection of Intrusions and Malware \&
  Vulnerability Assessment (DIMVA)}, 2017.

\bibitem{DBLP:conf/sp/ShokriSSS17}
Reza Shokri, Marco Stronati, Congzheng Song, and Vitaly Shmatikov.
\newblock Membership inference attacks against machine learning models.
\newblock In {\em IEEE Symposium on Security and Privacy (S\&P)}, 2017.

\bibitem{DBLP:conf/ndss/StefanovSS12}
Emil Stefanov, Elaine Shi, and Dawn~Xiaodong Song.
\newblock Towards practical oblivious {RAM}.
\newblock In {\em Symposium on Network and Distributed System Security (NDSS)},
  2012.

\bibitem{pathoram}
Emil Stefanov, Marten van Dijk, Elaine Shi, Christopher~W. Fletcher, Ling Ren,
  Xiangyao Yu, and Srinivas Devadas.
\newblock Path {ORAM:} an extremely simple oblivious {RAM} protocol.
\newblock In {\em ACM Conference on Computer and Communications Security
  (CCS)}, 2013.

\bibitem{DBLP:journals/corr/WaghCM16}
Sameer Wagh, Paul Cuff, and Prateek Mittal.
\newblock Differentially private oblivious {RAM}.
\newblock {\em PoPETs}, 2018(4):64--84, 2018.

\bibitem{wsc-bcomp-08}
Peter Williams, Radu Sion, and Bogdan Carbunar.
\newblock Building castles out of mud: practical access pattern privacy and
  correctness on untrusted storage.
\newblock In {\em ACM Conference on Computer and Communications Security
  (CCS)}, 2008.

\bibitem{sgxsidechannels}
Yuanzhong Xu, Weidong Cui, and Marcus Peinado.
\newblock Controlled-channel attacks: Deterministic side channels for untrusted
  operating systems.
\newblock In {\em IEEE Symposium on Security and Privacy (S\&P)}, 2015.

\bibitem{opaque}
Wenting Zheng, Ankur Dave, Jethro~G. Beekman, Raluca~Ada Popa, Joseph~E.
  Gonzalez, and Ion Stoica.
\newblock Opaque: An oblivious and encrypted distributed analytics platform.
\newblock In {\em USENIX Symposium on Networked Systems Design and
  Implementation (NSDI)}, 2017.

\end{thebibliography}
